\documentclass[a4paper,14pt]{article}

\usepackage{amsthm,amsmath,amssymb,amsfonts,mathrsfs}
\usepackage{graphicx}
\usepackage{bm}
\usepackage{geometry}

\newtheorem{proposition}{Proposition}

\newtheorem{theorem}{Theorem}

\newtheorem{lemma}{Lemma}

\newcommand{\diff}[2]{\frac{\partial{#1}}{\partial{#2}}}
\newcommand{\Alg}[0]{\mathscr{A}}

\begin{document}

	{\title{\protect\vspace*{-2cm}Symmetries and conservation laws \\ for a generalization of Kawahara equation}}
	\author{Jakub Va\v{s}\'\i{}\v{c}ek\\ Silesian University in Opava, Mathematical institute in Opava,\\
Na Rybn\'\i{}\v{c}ku 1, 746 01 Opava, Czech Republic\\ E-mail: {\tt jakub.vasicek@math.slu.cz}}

\maketitle
	\begin{abstract}
	We give a complete classification of generalized and formal symmetries and local conservation laws for a nonlinear evolution equation which generalizes the Kawahara equation having important applications in the study of plasma waves and capillary-gravity water waves.\looseness=-1
	
	In particular, we show that the equation under study admits no genuinely generalized symmetries and has only finitely many local conservation laws, and thus this equation is not symmetry integrable.
	\end{abstract}

\section{Introduction}
Consider the equation
\begin{equation}\label{kaw0}
u_t= \alpha uu_x + \beta u^2u_x+ \gamma u_{xxx}+ \mu u_{5x},
\end{equation}
where  $\alpha, \beta, \gamma$ and $\mu$ are constants, see \cite{kaw72}, as well as \cite{gan17, pol12} and references therein. This equation arises inter alia in the modeling of plasma waves and capillary-gravity waves, see \cite{gan17, kaw72} and references therein, and to this day remains a subject of intense research, see e.g.\ \cite{ dem10, gan17}. 

Equation \eqref{kaw0} is often referred to in the literature as the Kawahara equation, although the original equation from Kawahara's paper \cite{kaw72} has $\beta=0$. On the other hand, some authors refer to \eqref{kaw0} with $\alpha=0$ as to the modified Kawahara equation.

Below we consider the following natural generalization of \eqref{kaw0}: 	
\begin{equation}
\label{kaw}
	u_t= a u_{5x} + b u_{xxx} + f(u) u_x,
\end{equation}
where $a,b$ are constants, $a\neq 0$, and $f=f(u)$ is a nonconstant locally analytic function of $u$ only. For the sake of brevity we will refer to \eqref{kaw} as to the {\em generalized Kawahara equation (GKE)}. Upon rescaling of $t$ we can get rid of the coefficient $a$, so from now on we assume without loss of generality that $a=1$. Note that equation \eqref{kaw} includes as special cases some other known generalizations of \eqref{kaw0}, cf.\ e.g.\ \cite{dem10, pol12}.\looseness=-1 

It is well known that symmetries and conservation laws are very helpful in both the search for exact solutions and the numerical simulations \cite{olv93}.
On the other hand, integrable systems are well known to have rich symmetry algebras, cf.\ e.g.\ \cite{fok80, fok87, fok96, kvv18, mos, olv93, s17} and references therein. A simple observation that the computation of symmetries is, to a large extent, an algorithmic procedure, cf.\ e.g.\ \cite{olv93}, has, in a series of nontrivial developments, lead to symmetry-based integrability tests and the related notion of symmetry integrability. For a recent survey on the latter see \cite{mik09}, and for integrability in general and the transformations arising in the course of classification of integrable systems see e.g.\ \cite{cd, luc18, dun, fok80, fok87, kra11, kvv18, mik09, ps10, s18, vps} and references therein.\looseness=-1

The rest of the paper is organized as follows. In Section \ref{mr} we state our main results, that is, the full classification of generalized symmetries and local conservation laws for GKE \eqref{kaw} with $\partial f/\partial u\neq 0$. For the reader's convenience we recall the basic results from the symmetry integrability theory in Section~\ref{p}. In order to establish our main results, in Section~\ref{nfs} we prove nonexistence of formal symmetries of rank greater than twelve for GKE \eqref{kaw} with $\partial f/\partial u\neq 0$. Finally in Section \ref{gscl} we employ this nonexistence result and the Hamiltonian structure of GKE \eqref{kaw} to establish restrictions on the orders of generalized symmetries and local conservation laws the GKE \eqref{kaw} with $\partial f/\partial u\neq 0$ could possess.

\section{Main results}
\label{mr}	

The first of our results is that GKE \eqref{kaw} admits no genuinely generalized symmetries and thus is not symmetry integrable; all its generalized symmetries are described by the following theorem:

\begin{theorem}
\label{th1}
For an arbitrary function $f(u)$ such that $\partial f/\partial u\neq 0$ equation \eqref{kaw} admits only generalized symmetries which are equivalent to Lie point ones, namely,
\begin{enumerate}
\item[1)] if $f$ is such that $u \partial f / \partial u$, $\partial f / \partial u$ and 1 are not linearly dependent, then \eqref{kaw} has just two linearly independent generalized symmetries with the characteristics $Q_1=u_{5x}+ bu_{xxx} +fu_x$ and $Q_2=u_x$ corresponding to the translations in the time and space variables;

\item[2)] if $u \partial f / \partial u$, $\partial f / \partial u$ and 1 are linearly dependent, then 
 in addition to two generalized symmetries listed in 1), there is a generalized symmetry with the characteristic
\begin{enumerate}
\item $Q_3=t u_x + 1/\alpha$ if $\partial ^2 f / \partial u^2=0$, so $f=\alpha u + \beta$, where $\alpha, \beta$ are constants;
\item $Q_4=t u_x + (u+c)/\gamma$ if $\partial ^2 f / \partial u^2 \neq 0$, so $f=\gamma \ln(u+c) + \delta$, where $\gamma,\delta$ and $c$ are constants.
\end{enumerate}

\end{enumerate}
\end{theorem}

In this connection note that the authors of \cite{gan17} investigated the Lie point symmetries for the equation \begin{equation}
\label{kaw1}
	u_t= a(t) u_{5x} + b(t) u_{xxx} + c(t) f(u) u_x,
\end{equation} 
and the authors of \cite{kur14} did the same earlier for the special case when $f(u)=u^n$. 

While equation \eqref{kaw1}  is slightly more general than \eqref{kaw}, the generalized symmetries of either \eqref{kaw} or \eqref{kaw1} to the best of our knowledge were not studied in full generality in the earlier literature.
Although we considered only a special case, \eqref{kaw}, of \eqref{kaw1} we have obtained a significantly stronger result, namely a complete description of all generalized symmetries.

\begin{theorem}
\label{th2}
For an arbitrary function $f(u)$ such that $\partial f/\partial u\neq 0$ equation \eqref{kaw} admits only local conservation laws with the characteristics of order not greater than four, namely,
\begin{enumerate}
\item[1)] if $\partial ^2 f / \partial u^2 \neq 0$, then equation \eqref{kaw} has just three linearly independent local conservation laws with the conserved densities $\rho_1= u$, $\rho_2= u^2$ and $\rho_3=  (u_{xx}^2 - b u_x^2)/2 +  \widehat{r}$, where $\hat r(u)$ and $r(u)$ are defined by the formulas $\partial \hat r/ \partial u=r(u)$ and $\partial r/ \partial u=f(u)$, and the associated fluxes
\[ \begin{array}{rcl}
\sigma_1 &=& \displaystyle \frac{1}{2} \diff{f}{u} u^2+ u_{xx} b+ u_{4x},\\[3mm]
\sigma_2 &=& \displaystyle u u_{4x} - u_{xxx} u_x + \frac{1}{2} u_{xx}^2 +  bu u_{xx}-\frac{1}{2} b u_x^2 + \frac{1}{3} \diff{f}{u} u^3,\\[3mm]
\sigma_3 &=& \displaystyle -fbu_x^2+\diff{f}{u}u_x^2u_{xx}-b^2u_xu_{xxx}+\frac{1}{2}b^2u_{xx}^2+rbu_{xx}-fu_xu_{xxx}+fu_{xx}^2\\[3mm]
&& \displaystyle +2bu_{4x}u_{xx}-bu_{5x}u_x-bu_{xxx}^2 + \frac{1}{2}r^2+ru_{4x}+\frac{1}{2}u_{4x}^2-u_{5x}u_{xxx}+u_{xx}u_{6x}; \end{array} \] \vspace*{-0.2cm}

\item[2)] if $\partial ^2 f / \partial u^2= 0$, then $f=\alpha u + \beta$, where $\alpha, \beta$ are constants, and \eqref{kaw} admits, in addition to the local conservation laws listed in 1), a local conservation law with the conserved density $\rho_4 = x u + \alpha t u^2/2$ and the associated flux
\[ \begin{array}{rcl}\sigma_4 &=& \displaystyle \frac{1}{6}\alpha((-3bu_x^2+6buu_{xx}+6u_{4x} u - 6u_{xxx} u_x + 3u_{xx}^2)t+3xu^2) + \frac{1}{2}\alpha^2 t u^3\\[3mm]
&&  + 3 b(xu_{xx}-u_x)+xu_{4x}-u_{xxx}. \end{array} \] \vspace*{-0.2cm}
\end{enumerate}
\end{theorem}

While these conservation laws were already found in \cite{gan17} as a result of the search for local conservation laws with the characteristics of order up to four, we go significantly further: namely, we prove that \eqref{kaw} has no other local conservation laws (with the characteristics of arbitrarily high order) whatsoever.

\section{Preliminaries}
\label{p}
In this section, we shall briefly recall a number of known results following mostly \cite{fok80, fok87, kup00, mik87, mik09, olv93, ps10, s99, sev16, sok88}.

We consider an evolution equation in two independent and one dependent variable of the form
\begin{equation}
\label{evo}
u_t=K(x,u,u_x,\dots,u_{nx}), \qquad n \geq 2, 
\end{equation}
where $K$ is locally analytic with respect to all its arguments.
For partial derivatives in $x$ we employ the usual abbreviated notation, for example $\partial^3 u/\partial x^3=u_{3x}=u_{xxx}$, $\partial^4 u/\partial x^4=u_{4x}$, and so on, including $u_{0x}\equiv u$.

Furthermore let $D_x$ and $D_t$ be total derivatives in $x$ and $t$ restricted to \eqref{evo}, that is:
$$ D_x=\diff{}{x} + \sum_{i=0}^\infty{u_{(i+1)x}}\diff{}{u_{ix}}, \qquad D_t=\diff{}{t} + \sum_{i=0}^\infty{D_x^i(K)\diff{}{u_{ix}}}$$

Recall \cite{kra11, kvv18} that a {\em local function} is a smooth function that may depend on $x,t,u,u_x,\dots,\allowbreak u_{kx}$
for an arbitrary but finite $k$.

A function $f(x,t,u,\dots,u_{sx})$ is called a {\em rational local function} if and only if it can be written as $$f=\frac{g}{h},$$ where $g,h$ are local functions polynomial in all their arguments. We shall denote the (differential) field of rational local functions by $\mathscr{A}_0$, cf.\ e.g.\ \cite{des12, rit66}.

Let $\mathscr{A}$ be a differential extension field of $\mathscr{A}_0$ such that $K \in \Alg$, so $\Alg$ is closed under $D_x$ and $D_t$.
In what follows we will refer to the elements of $\Alg$ as to the \emph{differential functions}. Unless explicitly stated otherwise all functions below are assumed to belong to $\Alg$.

An \emph{evolutionary vector field}
 with the characteristic $Q$ from $\Alg$ has the form
$$\mathbf{v}_Q=Q\diff{}{u}.$$
Such a vector field $\mathbf{v}_Q$ is called a \emph{generalized symmetry} for \eqref{evo} iff its \emph{characteristic} $Q$ satisfies \looseness-1
\begin{equation}
\label{chsym}
D_t(Q) =\mathrm{\textbf{D}}_K(Q).
\end{equation}
Here, for $F=F(x,t,u,u_x,\dots,u_{kx})\in\Alg$ we denote by $\mathrm{\textbf{D}}_F$ the (formal) Fr\'echet derivative
\begin{equation}
\mathrm{\textbf{D}}_F = \sum_{j=0}^{k}{\diff{F}{u_{jx}} D_x^j}.
\end{equation}

Following \cite{kup00} consider an algebra $\mathscr{L}$ of formal series in $\xi$ of the form

\begin{equation}
L=\sum_{i=-\infty}^k a_i\xi^i,
\end{equation}
where the coefficients $a_i$ are from the differential field $\mathscr{A}$. The degree of such a formal series is defined as the greatest $j$ such that $a_j \neq 0$, with the convention that $\deg 0=-\infty$.

The multiplication of two monomials is defined by the formula
$$a\xi^i\circ b\xi^j=a\sum_{k=0}^\infty{\frac{i(i-1)\cdot\cdot\cdot(i-k+1)}{k!}D_x^k(b)\xi^{i+j-k}},$$
which by virtue of linearity is extended to the whole $\mathscr{L}$. Such a multiplication can be shown to be associative which implies that the commutator $\left[P,Q\right]=P\circ Q -Q\circ P$ endows $\mathscr{L}$ with a Lie algebra structure. To simplify writing, in what follows $\circ$ would be omitted whenever there is no risk of confusion.

Next, for any $L \in \mathscr{L}$
define its formal adjoint $L^*$ as follows:

\begin{equation}
\mathrm{if} \quad L=\sum_{i=-\infty}^q{a_i\xi^i}, \quad \mathrm{then} \quad L^*=\sum_{i=-\infty}^q{(-\xi)^i \circ a_i}.
\end{equation}

It is well known that every $L \in \mathscr{L}$ with  $\deg L=n>0$ has its $n$-th root which satisfies $$L=\underbrace{{L}^{1/n}\circ {L}^{1/n}\circ \dots\circ {L}^{1/n}}_{n\emph{~times}}=(L^{1/n})^n.$$

A \emph{formal symmetry of rank} $k$ for \eqref{evo} is a formal series $L \in \mathscr{L}$ of degree $m$ which satisfies
\begin{equation}
\label{fs}
\mathrm{deg}(D_t(L)-[\widehat{\emph{\textbf{D}}}_K , L]) \leq m+n-k,
\end{equation}
where for any $F=F(x,t,u,u_x,\dots,u_{kx})\in\Alg$ we define
\[
\widehat{\emph{\textbf{D}}}_F = \sum_{j=0}^{k}{\diff{F}{u_{jx}} \xi^j}.
\]

For any differential function $F\in\Alg$ we define its \emph{order} $\mathrm{ord}\ F$ as $\deg \widehat{\mathrm{\textbf{D}}}_F$.

Without loss of generality any \emph{local conservation law} for \eqref{evo} can be assumed to take the form,
\begin{equation}
\label{con}
D_t(\rho)=D_x(\sigma),
\end{equation}
where 
$\rho \in \mathscr{A}$ is called a \emph{conserved density} 
and $\sigma \in \Alg $ is the associated \emph{flux}. We shall omit the word \emph{local}, since we will deal only with local conservation laws, densities and fluxes.
The characteristic of a conservation law (\ref{con}) for \eqref{evo} is a differential function $P \in \Alg$ which satisfies
$$D_t(\rho)-D_x(\sigma)=P\cdot (u_t-K).$$

Two conservation laws for \eqref{evo}, $D_t(\rho)=D_x(\sigma)$ and $D_t(\tilde\rho)=D_x(\tilde\sigma)$ are \emph{equivalent}, see e.g. \cite{olv93, ps10}, if there exists a differential function $\zeta\in\Alg$ such that $\tilde\rho-\rho=D_x(\zeta)$ and $\tilde\sigma-\sigma=D_t(\zeta)$. It is readily seen that two equivalent conservation laws share the same characteristic.
Below we shall tacitly assume that the conservation laws are considered modulo the above equivalence.

An evolutionary partial differential equation \eqref{evo} is said to be \emph{Hamiltonian} if it can be written in the form
\begin{equation}
\label{ham}
\frac{\partial{u}}{\partial{t}}=\mathscr{D}\delta\mathscr{H},
\end{equation}
where $\mathscr{D}$ is a Hamiltonian operator, 
see e.g.\ Ch.\ 7 of \cite{olv93} for details,
$\delta$ is the operator of variational derivative, and $\mathscr{H}=\int{H}\, dx$, where $H \in \Alg$, is usually referred to as the \emph{Hamiltonian functional}, or just the \emph{Hamiltonian}.

Recall that the operator of variational derivative $\delta \mathscr{H}$
is defined as follows:
\begin{equation}
\delta \mathscr{H}=\frac{\delta }{\delta u} H = \sum_{i=1}^\infty{(-D_x)^i\biggl(\diff{H}{u_i}}\biggr).
\end{equation}

\begin{lemma}\emph{(\cite{olv93})}
\label{prop}
Consider a Hamiltonian equation in the form (\ref{ham}) and let (\ref{con}) define a conservation law. Then $\mathscr{D}\left(\delta\rho/\delta u\right)$ is a characteristic of a generalized symmetry of this equation.
\end{lemma}

\section{Nonexistence of formal symmetries}
\label{nfs}

As we shall see in the next section, Theorems \ref{th1} and \ref{th2} 
are consequences of the following result:

\begin{theorem}
\label{th3}
If $b \neq 0$ and $f(u)$ is such that $\partial f/\partial u\neq 0$, then GKE \emph{(\ref{kaw})} has no nontrivial formal symmetry of rank 13 or greater.
\end{theorem}

\begin{proof}[Proof of Theorem \ref{th3}]

We will prove Theorem \ref{th3} by contradiction, so suppose that there exists a nonconstant $L \in \mathscr{L}$ with $\deg L \neq 0$, which is a formal symmetry of rank 13 (or greater). If so, this $L$ must satisfy equation (\ref{fs}), namely:
\begin{equation}
\label{fsk}
\mathrm{deg}\, (D_t(L)-[\widehat{\mathrm{\textbf{D}}}_K , L]) \leq \mathrm{deg}\, L + \mathrm{deg}\, \widehat{\mathrm{\textbf{D}}}_K -13.
\end{equation}

Without loss of generality (cf.\ \cite{mik87, mik09, olv93, s99}) we set $\deg L=1$, so $L$ takes the form

$$L=g\xi+\sum_{i=0}^\infty{l_i\xi^{-i}},$$
with $g,l_i \in \mathscr{A}$. Then equation (\ref{fsk}) boils down to
\begin{equation}
\mathrm{deg}(D_t(L)-[\widehat{\mathrm{\textbf{D}}}_K , L]) \leq -7.
\end{equation}

We now need to equate to zero the coefficients at $\xi^i$ step by step, starting from $i=5$ (it is readily seen that the coefficients at $\xi^i$ for $i>5$ are identically zeroes) and going down.

For $i=5$ we get $$-5D_x(g)=0.$$ This means that $g$ is an arbitrary  function of $t$ only, as the kernel of $D_x$ is exhausted by such functions \cite{olv93}.

Next we equate to zero the coefficient at $\xi^4$ which yields
 $$-5D_x(l_0)=0,$$
 so likewise $l_0$ is a function of $t$ alone.

If we continue like this for $i=4,3,2$ we obtain equations of the same shape as above, $$-5D_x(l_i)=0,$$ which are solved in the same way as the one for $g$.

For $i=1,0,-1,\dots$ we get slightly more complicated equations, namely
\begin{equation}
\label{coe}
-5D_x(l_i)=F_i
\end{equation}
for some $F_i \in \Alg$, which for $i=1,0,-1,-2$ are however solved just as easily as above. Recall \cite{olv93} that a necessary condition for this kind of equations to be solvable in the class of differential functions is that the equality $\delta F_i/\delta u=0$ holds.

The first case when the condition $\delta F_i/\delta u=0$ is nontrivial occurs for $i=-3$.

It can be readily checked that for $\delta F_{-3}/\delta u=0$ to hold we have to solve the following system:
\begin{equation}
\label{nont}
-\frac{1}{5} g \frac{\partial ^3 f}{\partial u^3}=0 , \qquad \frac{3}{25}
\frac{\partial f}{\partial u}  \frac{\partial g}{\partial t}  =0.
\end{equation}

The first equation tells us that if $\partial ^3 f/\partial u^3 \neq 0$ we arrive at a contradiction with our initial assumption, because in this case $g$ would have to be zero, i.e., for the case when $\partial ^3 f/\partial u^3 \neq 0$ our theorem is already proved.

We shall now continue with the case when  $\partial^3 f/ \partial u^3=0$, i.e., $f$ is at most quadratic polynomial in $u$. Using the transformation $u\longrightarrow u-p_1/(2p_2)$ we can turn any polynomial $f=p_2u^2+p_1u+p_0$ into $\widetilde{f}(u)=p_2u^2+\widetilde{p_0}$. Quite similarly we can eliminate $\widetilde{p_0}$ by transformation $x \longrightarrow x+\widetilde{p_0}t$ and by rescaling $u \longrightarrow u/\sqrt{p_2}$ to get rid of $p_2$. Thus without loss of generality we are left with the case $$f(u)=u^2.$$

Then, if we take a closer look at the second equation in \eqref{nont}, we see that
\begin{equation}\label{fceg}\diff{g}{t}=0,\end{equation}
which can be satisfied only if $g$ is a constant. 

For $i=-4$ the condition $\delta F_{-4} /\delta u=0$ yields
$$\frac{6}{25} u \frac{\partial  l_0}{\partial t}=0,$$
so $l_0$ turns out to be a constant rather than a function of $t$.

However, we know that constants are trivial formal symmetries for any equation \eqref{evo}, so to simplify our computations we put $l_0=0$ without loss of generality.

For $i=-5, -6$ we have  $\partial l_i/\partial t = 0$, so $l_1$ and $l_2$ are arbitrary constants just like $g$.

Finally for $i=-7$ we obtain the system which contains, among other, the equation
$$g=0,$$
which 
contradicts the initial assumption that deg\,$L=1$ so the proof is completed.
\end{proof}

\section{Generalized symmetries and conservation laws}
\label{gscl}
We shall need the following lemma:
\begin{lemma}
\label{lem}
If $G$ is a characteristic of generalized symmetry of order $s$ for GKE \eqref{kaw} then $\widehat{\textbf{D}}_G$
is a formal symmetry of degree $s$ and rank at least $s+4$ for this equation.
\end{lemma}

\begin{proof}
Indeed, suppose that $G$ is a characteristic of a generalized symmetry, then by virtue of $D_t(G)=\mathrm{\textbf{D}}_K(G)$, where $K$ now denotes the right-hand side of \eqref{kaw}, we have \cite{sok88}
\[
D_t(\widehat{\mathrm{\textbf{D}}}_G)-[\widehat{\mathrm{\textbf{D}}}_K,\widehat{\mathrm{\textbf{D}}}_G]-\mathrm{\textbf{D}}_G(\widehat{\mathrm{\textbf{D}}}_K)=0,
\]
so
\[
\deg\left(D_t(\widehat{\mathrm{\textbf{D}}}_G)-[\widehat{\mathrm{\textbf{D}}}_K,\widehat{\mathrm{\textbf{D}}}_G]\right)=\deg\left(\mathrm{\textbf{D}}_G(\widehat{\mathrm{\textbf{D}}}_K)\right)
\]
but as in our case
\[
\widehat{\mathrm{\textbf{D}}}_K=\xi^5+b\xi^3+f\xi+\diff{f}{u}u_x,
\]
we have
\[
\deg\left(\mathrm{\textbf{D}}_G(\widehat{\mathrm{\textbf{D}}}_K)\right)\leq 1,
\]
and we see that $\widehat{\mathrm{\textbf{D}}}_G$ indeed is a formal symmetry for \eqref{kaw} of rank (at least) $s+4$, where $s$ is the order of $G$.
\end{proof}
By the above lemma Theorem \ref{th3} implies that GKE has no generalized symmetries of order greater than 8. All generalized symmetries of order up to 8 for \eqref{kaw} can be readily computed, and, in particular, we arrive at the following result:
\begin{proposition}\label{psym}
GKE with $f(u)$ such that $\partial f/\partial u\neq 0$ admits only generalized symmetries which are equivalent to Lie point ones, i.e., it has no genuinely generalized symmetries.
\end{proposition}

The subsequent computation of Lie point symmetries for GKE \eqref{kaw} gives rise to Theorem \ref{th1}.

Next, using the Hamiltonian structure we obtain the following

 \begin{proposition}
\label{consth}
GKE \eqref{kaw} with $f(u)$  such that $\partial f/\partial u\neq 0$ admits only local conservation laws with the characteristics of order not greater than four.
\end{proposition}
\emph{Sketch of proof.}
Recall, see e.g.\ \cite{gan17}, that \eqref{kaw} admits a Hamiltonian operator $\mathscr{D}=D_x.$ Using Lemma \ref{prop} we see that applying the operator $\mathscr{D}$ to a characteristic $P$ of a conservation law yields a characteristic of a symmetry of order higher by one than that of $P$, so we have to consider only conservation laws with order one less than the greatest order of previously found symmetries which by Proposition~\ref{psym} is five; cf.\ \cite{vod} for a similar argument. Thus, the most general conservation law for \eqref{kaw} has a characteristic of the form $P=P(x,t,u,u_x,u_{xx},u_{xxx},u_{4x})$.
\hfill $\Box $
\looseness-1

Using Proposition \ref{consth} we can readily find all local conservation laws for GKE \eqref{kaw} which we presented in Theorem \ref{th2}.

	\subsection*{Acknowledgments}
I would like to express my most sincere gratitude to Artur Sergyeyev for stimulating discussions and valuable comments.

A large part of the computations in the article were performed using the package \emph{Jets} \cite{jets}.

This research was supported by the Specific Research grant SGS/6/2017 of the Silesian University in Opava. \nopagebreak[4]
	


\small\begin{thebibliography}{11}
	\itemsep=-0.25mm
	\footnotesize
\bibitem{jets}
  Baran, H.; Marvan, M.,
	\emph{Jets. A software for differential calculus on jet spaces and diffieties},
	available online at {\tt{http://jets.math.slu.cz}}	




\bibitem{cd}
Calogero, F.; Degasperis, A.,
 \emph{Spectral transform and solitons. Vol. I. Tools to solve and investigate nonlinear evolution equations},
    North-Holland Publishing Co., Amsterdam-New York, 1982.

\bibitem{luc18}
 de Lucas, J.; Grundland, A.M., 
 \emph{A cohomological approach to immersed submanifolds via integrable systems}, 
 Selecta Math. (N.S.) 24 (2018), no. 5, 4749--4780.
	
\bibitem{des12}
 De Sole, A.; Kac, V.G.,
 \emph{Essential variational Poisson cohomology},
 Comm. Math. Phys. 313 (2012), no. 3, 837--864, 	arXiv:1106.5882.
		
\bibitem{dem10}
 Demina, M.V.;  Kudryashov, N.A.,
\emph{From Laurent series to exact meromorphic solutions: The Kawahara equation},
Phys. Lett. A 374 (2010), no. 39, 4023--4029, arXiv:1112.5266.

\bibitem{dun}Dunajski, M., \emph{Solitons, Instantons and Twistors}, Oxford University Press, 2009.

\bibitem{fok80}
Fokas, A.S.,
 \emph{A symmetry approach to exactly solvable evolution equations},
 J. Math. Phys. 21 (1980), no. 6, 1318--1325.

\bibitem{fok87}
Fokas, A.S.,
\emph{Symmetries and integrability},
Stud. Appl. Math. 77 (1987), no. 3, 253--299. 

\bibitem{fok96}
 Fokas, A.S.; Gelfand, I.M.,
 \emph{Surfaces on Lie groups, on Lie algebras, and their integrability},
 Commun. Math. Phys. 177 (1996), no. 1, 203--220.

\bibitem{gan17}
Gandarias, M.L.; Rosa, M.; Recio, E.; Anco, S.,
	\emph{Conservation laws and symmetries of a generalized Kawahara equation},
   AIP Conference Proceedings 1836 (2017), no. 1, 020072.

\bibitem{kaw72}
Kawahara, T.,
\emph{Oscillatory Solitary Waves in Dispersive Media},
J.\ Phys.\ Soc.\ Jpn.\ 33 (1972), 260--264. \looseness-1


\bibitem{kra11}
  Krasil'shchik, J.; Verbovetsky, A.,
  \emph{Geometry of jet spaces and integrable systems},
  J. Geom. Phys. 61 (2011), no. 9, 1633--1674,	arXiv:1002.0077.
	
\bibitem{kvv18}
Krasil'shchik, J.; Verbovetsky, A.; Vitolo, R.,
  \emph{The symbolic computation of integrability structures for partial differential equations},
	Springer, Cham, 2017.
	
\bibitem{kup00}
	Kupershmidt, B.A.,
	\emph{KP or mKP. Noncommutative mathematics of Lagrangian, Hamiltonian, and integrable systems.}
	Mathematical Surveys and Monographs, 78. AMS, Providence, RI, 2000.

\bibitem{kur14}
  Kuriksha, O.; Po\v{s}ta, S.; Vaneeva, O.,  
\emph{Group classification of variable coefficient generalized Kawahara equations}
	 J. Phys. A: Math. Theor. 47 (2014), arXiv:1309.7161. 


\bibitem{mik87}	
  Mikhailov, A.V.; Shabat, A.B.; Yamilov, R.I.,
  \emph{The symmetry approach to the classification of non-linear equations. Complete lists of integrable systems},
  Russ.\ Math.\ Surv. 42 (1987), no. 4, 1--63.
	
\bibitem{mik09}
  Mikhailov, A.V.; Sokolov, V.V.,
	\emph{Symmetries of Differential Equations and the Problem of Integrability}, in \emph{Integrability}, ed. Mikhailov A.V.,
  Springer, Berlin Heidelberg, 2009, 19--88.

\bibitem{mos}
Morozov, O.I.; Sergyeyev, A.,
\emph{The four-dimensional Mart\'\i{}nez Alonso--Shabat equation: reductions and nonlocal symmetries}, 
J. Geom. Phys. 85 (2014), 40--45, arXiv:1401.7942.

\bibitem{olv93}
  Olver, P.J.,
  \emph{Applications of Lie Groups to Differential Equations},
  Springer, New York,
  2nd ed.,
  1993.\looseness=-1
	
\bibitem{pol12}
  Polyanin, A.D.; Zaitsev, V.F.,
	\emph{Handbook of nonlinear partial differential equations},
	CRC Press, Boca Raton, FL,
	Second edition,
	2012.
	
\bibitem{ps10}
Popovych, R.O.; Sergyeyev, A.,
\emph{Conservation laws and normal forms of evolution equations},
Phys. Lett. A 374 (2010), no. 22, 2210--2217,  	arXiv:1003.1648.	

\bibitem{rit66}
 Ritt, J.F., \emph{Differential algebra}, Dover Publications, Inc., New York 1966.

\bibitem{s99}
Sergyeyev, A.,
\emph{On time-dependent symmetries and formal symmetries of evolution equations},
in \emph{Symmetry and Perturbation Theory} (Rome, 1998), ed.\ G. Gaeta, World Scientific,
1999, 303--308, arXiv:solv-int/9902002.

\bibitem{s17}
Sergyeyev, A.,
\emph{A Simple Construction of Recursion Operators for Multidimensional Dispersionless Integrable Systems},
 J. Math. Anal. Appl. 454 (2017), no. 2, 468--480, arXiv:1501.01955.

\bibitem{s18}
Sergyeyev, A.,
\emph{New integrable (3+1)-dimensional systems and contact geometry},
Lett. Math. Phys. 108 (2018), no. 2, 359--376, arXiv:1401.2122.



\bibitem{sev16}
Sergyeyev, A.; Vitolo, R.,
  \emph{Symmetries and conservation laws for the Karczewska--Rozmej--Rutkowski--Infeld equation},
  Nonlinear Anal. Real World Appl. 32 (2016), 1--9, arXiv:1511.03975.


\bibitem{sok88}
 Sokolov, V.V.,
\emph{On the symmetries of evolution equations},
Russ.\ Math.\ Surv.\ 43 (1988), no.\ 5, 
165--204.\looseness=-1

\bibitem{vps}
Vaneeva, O.O.; Popovych, R.O.; and Sophocleous, C.,
\emph{Equivalence transformations in the study of integrability},
Phys. Scr. 89 (2014) 038003, 9 p., arXiv:1308.5126.

\bibitem{vod}
Vodov\'a, J.,
\emph{A complete list of conservation laws for non-integrable compacton equations of $K(m,m)$ type},
Nonlinearity 26 (2013), no. 3, 757--762,  arXiv:1206.4401.
	
\end{thebibliography}
\end{document}